\documentclass[ reqno,12pt]{article}
\usepackage{eurosym}
\usepackage{amsfonts}
\usepackage{amsmath}
\usepackage{amssymb}
\usepackage{amsthm}
\usepackage{mathrsfs}
\usepackage{mathtools}
\usepackage{gensymb}
\usepackage{upgreek}
\usepackage{relsize}
\usepackage[margin=1in]{geometry}
\usepackage{hyperref}
\usepackage{ragged2e}
\usepackage{pdflscape}
\usepackage[section]{placeins}
\usepackage[nodisplayskipstretch]{setspace}
\usepackage{verbatim}
\usepackage[justification=centering]{caption}
\usepackage{epstopdf,epsfig}
\usepackage{graphicx}
\usepackage{subcaption}
\usepackage{tikz}
\usepackage{array}
\usepackage{arydshln}
\usepackage{booktabs}
\usepackage{multirow}
\usepackage{tabularx}
\usepackage{graphicx,amsmath,amssymb,epsfig}
\usepackage{amsfonts}
\usepackage{amsthm}
\usepackage{booktabs}
\usepackage{subcaption}

\hypersetup{
colorlinks=true, linkcolor=black, citecolor=black}
\hypersetup{urlcolor=gray}
\usetikzlibrary{arrows,decorations.pathreplacing, decorations.markings, matrix, positioning, patterns}
\usetikzlibrary{calc}
\newtheorem{theorem}{Theorem}
\theoremstyle{plain}

\newtheorem{condition}{Assumption}

\newtheorem{definition}{Definition}

\newtheorem{lemma}{Lemma}

\newtheorem{proposition}{Proposition}[section]
\newtheorem*{proposition*}{Proposition}
\theoremstyle{definition}
\newtheorem{remark}{Remark}[section]

\numberwithin{equation}{section}

\newenvironment{note}
    {\bgroup\footnotesize  \begin{justify}
    }
    {\end{justify}
    \egroup }

\newenvironment{petit}{\footnotesize}

\input{tcilatex}

\begin{document}

\title{\vspace{-40pt}A Note on the Estimation of Job Amenities and Labor
Productivity\thanks{%
We thank two anonymous referees and the editor Chris Taber for insightful
comments.}}
\author{Arnaud Dupuy\thanks{%
CREA, University of Luxembourg and IZA, \href{mailto:arnaud.dupuy@uni.lu}{%
arnaud.dupuy@uni.lu}. Dupuy gratefully acknowledges the support of a FNR
grant C14/SC/8337045.} \and Alfred Galichon\thanks{%
Departments of Economics and of Mathematics, New York University, and Department of Economics, Sciences-Po, \href{mailto: ag133@nyu.edu}%
{ag133@nyu.edu}. Galichon gratefully acknowledges the support of NSF grant
DMS-1716489 and European Research Council Grant no. 866274.}}
\date{First circulated version: October, 2015. This version: July, 2021}
\maketitle

\begin{abstract}
This paper introduces a maximum likelihood estimator of the value of job
amenities and labor productivity in a single matching market based on the
observation of equilibrium matches and wages. The estimation procedure
simultaneously fits both the matching patterns and the wage curve. While our
estimator is suited for a wide range of assignment problems, we provide an
application to the estimation of the Value of a Statistical Life using
compensating wage differentials for the risk of fatal injury on the job.
Using US data for 2017, we estimate the Value of Statistical Life at \$%
   6.3million (\$2017).
\end{abstract}

\vspace{-15pt}

\indent{\footnotesize \textbf{Keywords}: Matching, Observed transfers, Structural
estimation, Value of Statistical Life.}

\indent{\footnotesize \textbf{JEL Classification}: C35, C78, J31.\vskip50pt }

\newpage

\section{Introduction}

Identification and estimation of both agents' value of a match in one-to-one
matching models with transferable utility have been the subject of
increasing interest in the last decade. Two important applications are in
the fields of family economics with the marriage market (where the
econometrician observes matching patterns, but not the transfers) and labor
economics with the labor market, or more generally the literature on hedonic
models (where the econometrician observes both the matching patterns and the
transfers), although the former has thus far received most of the attention.

In the case when transfers are not observed, thus in the case of the
marriage literature, Choo and Siow (2006) is a seminal reference which
allowed to bring theoretical models to the data. Subsequent references such
as Chiappori et al. (2015), Galichon and Salani\'{e} (2021) and Dupuy and
Galichon (2014) have extended the structure of the model in various
dimensions. In particular, Dupuy and Galichon (2014) have provided a
framework for estimation of a matching model where agents match on
continuous characteristics, which they have applied to marriage market data.

In the case when transfers are observed, however, transfers may potentially
provide useful supplementary information about the partners' values of a
match. In the analysis of the labor market, for example, wages may be
observed. The literature referred to above is not very explicit on how this
information may be used. Many authors, such as Ekeland et al. (2004),
Heckman et al. (2010) and Galichon and Salani\'{e} (2021), among others,
suggest techniques that implicitly or explicitly require to perform
nonparametric estimation (\textquotedblleft hedonic
regression\textquotedblright ) of the wage curve prior to the analysis.
While this works well in the case when the relevant characteristics is
single-dimensional, as in Ekeland et al. (2004) and Heckman et al. (2010),
or discrete, as in Galichon and Salani\'{e} (2021), this is more involved
when the characteristics are continuous and multivariate. In this framework,
Salani\'{e} (2015) shows that this structure implies quite strong testable
restrictions.

In this paper, we build a flexible and tractable model of equilibrium
matching and wages on the labor market, and show how to estimate the model
using a maximum likelihood approach. This work therefore extends our
previous work, Dupuy and Galichon (2014), to the case when transfers are
observed.

We illustrate our method by revisiting the literature on compensating wage
differentials (CDW) initiated as an application of Rosen's (1974) hedonic
model to the labor market by Lucas (1977) and Thaler and Rosen (1976), and
soon followed by many others; see Rosen (1986) for an elegant presentation
of the theory and a review of the early empirical literature, and Viscusi
and Aldy (2003) for a more recent review of the empirical literature. The
approach in this vein consists in performing the reduced form estimation of
the risk-wage gradient to uncover workers' marginal willingness to accept
certain levels of fatal injury risk at their job and herewith derive an
estimate of the Value of Statistical Life. A crucial assumption of this
approach is that the data contains rich enough information about a worker's
skills to control for wage differentials due to productivity differentials
across workers. Departure from this assumption implies an inherent bias in
estimates of the compensating wage differentials, and Hwang et al. (1992)
have shown this bias can be large in magnitude. Attempts to avoid this bias
have consisted in either using panel data to estimate workers' fixed effects
and control for unobserved heterogeneity (see for instance Brown, 1980) or
an instrumental variable approach (see for instance Garen, 1988). More
recently, Kniesner et al. (2007) have argued in favor of adopting a
structural hedonic model to identify the ``underlying fundamentals
(preferences), [...] that would further generalize estimates of Value of
Statistical Life''.

To the extent of our knowledge, however, the structural hedonic model
approach has been largely ignored in the applied literature. Our method
contributes to this discussion by proposing a structural estimation of
preferences for risky jobs that explicitly takes into account the matching
of workers to jobs while estimating the equilibrium hedonic wage equation.
This method comes at no cost on the data since information about who matches
with whom is, by definition, already available in the data needed to perform
the hedonic wage regression in the first place. Accounting for the matching
of workers to jobs results in the likelihood of observing the data given
parameters being expressed as a weighted sum of two terms: the contribution
of the first term is to equate the predicted moments of the matching
distributions to their sample counterparts whereas the contribution of the
second is to equate the predicted wages with their sample counterparts.

Following Viscusi (2003, 2007 and 2013), we use US data and merge the Census
of Fatal Occupational Injuries (CFOI) by occupation and industry to the 2017
Current Population Survey (CPS). This allows us to have access to data on
hourly wages, workers' characteristics and the rate of fatal injuries in
their job. Our main results quantify the extent to which US workers dislike
risky jobs, their utility dropping by  0.023
log-points per hour of work as the probability of fatal injury on the job
increases by  13.05per $100,000$. This amounts to a Value
of Statistical Life (VSL hereafter) of \$   6.3million
(\$2017). This estimate is about \$3 million lower, though not statistically
so, than the estimate obtained when applying a hedonic (log)wage regression
that does not account explicitly for the sorting of workers into jobs, i.e.
\$9.7 million (\$2017), which itself lies in the range of previous estimates
using similar data (e.g. Viscusi, 2013).

Our model is also related to a growing empirical literature applying the
celebrated estimation technique proposed in Abowd et al. (1999) to decompose
workers' wage differentials into differentials due to observed workers'
characteristics, unobserved workers' heterogeneity and firms' heterogeneity
using matched employer-employee panel data, see among others Abowd et al.
(2002), Andrews et al. (2008;2012), Gruetter and Lalive (2009), Woodcock
(2010) and Torres et al. (2013). Workers and firms fixed effects capture
reduced form notions of workers and firms types that are fixed over time and
are identified using the mobility of workers across firms over time.

The outline for the rest of the paper is as follows. Section 2 introduces
the model and characterizes equilibrium. Section 3 presents our parametric
specification of the model and a maximum likelihood estimator on data about
matches and wages. Section 4 presents the empirical application and section
5 summarizes and concludes.

\section{The model}

The purpose of this section is to succinctly present our model, which is a
bipartite continuous matching model with transferable utility and logit
unobserved heterogeneity. In our context, equilibrium transfers (wages) are
observed. This is relevant for instance in the labor market, as opposed to
the marriage market where transfers are typically unobserved. We limit
ourselves to the introduction of the notation needed for the construction of
our estimator, emphasize on the additional identification and estimation
results obtained when transfers are observed and refer the interested reader
to the original paper, i.e. Dupuy and Galichon (2014), for more details
about its other main properties. To fix
ideas, we use in the remainder of the paper the example of the labor market where transfers (wages) are observed, in line with the application of Section 4.

\textbf{Populations and matching}. We shall assume that workers'
characteristics are contained in a vector of attributes $x\in \mathcal{X}=%
\mathbb{R}^{d_{x}}$, while firms' characteristics are captured by a vector
of attributes $y\in \mathcal{Y}=\mathbb{R}^{d_{y}}$.\footnote{\label%
{fnt:OneToMany}In our one-to-one matching model, the terms \textquotedblleft
job\textquotedblright\ and \textquotedblleft firm\textquotedblright\ are
interchangeable. Our model would continue to work in a one-to-many context
where firms offer multiple jobs as long as there is perfect substitutability
between jobs (workers) within firms, i.e. as long as the surplus of a firm
is the sum of the surplus at each job (of each worker) in the firm. Our
model can therefore be seen as a matching model of jobs to workers within
firms under perfect substitutability. A recent application of this idea is
found in a model of polygamy in the marriage market context by Andr\'{e} and
Dupraz (2017).}

Our first main assumption is about the distribution of workers' and firms'
types in the economy.

\begin{condition}\label{Ass1} There is a continuum of workers, with a density of
type distribution $f$ on $\mathbb{R}^{d_{x}}$, and a continuum of firms,
with a density of type distribution $g$ on $\mathbb{R}^{d_{y}}$. There is
the same total mass of workers and firms, and this mass is normalized to
one, hence%
\begin{equation*}
\int_{\mathbb{R}^{d_{x}}}f\left( x\right) dx=\int_{\mathbb{R}%
^{d_{y}}}g\left( y\right) dy=1.
\end{equation*}
\end{condition}

Since workers of type $x$ have a density of probability $f\left( x\right) $,
and firms of type $y$ have a density of probability $g\left( y\right) $ and
workers and firms are in equal number, a \emph{feasible matching} between
workers and firms will consist in the probability density $\pi \left(
x,y\right) $ of occurrence of a $\left( x,y\right) $ pair, which should have
marginal densities $f$ and $g$. More formally, we define the set of feasible matchings as%
\begin{equation*}
\mathcal{M}\left( f,g\right) =\left\{ \pi :\pi \left( x,y\right) \geq 0\text{%
, }\int_{\mathcal{Y}}\pi \left( x,y\right) dy=f\left( x\right) \text{ and }%
\int_{\mathcal{X}}\pi \left( x,y\right) dx=g\left( y\right) \right\} .
\end{equation*}

\bigskip

\textbf{Demand and supply}. Let $w\left( x,y\right) $ denote the wage of a
worker of type $x$ when working for a firm of type $y$. It is assumed that a
worker of type $x$ not only values her wage but also
the amenities of her job. The value of amenities is
further assumed to be decomposed into a systematic value $\alpha \left(
x,y\right) $, which is the same for all workers of type $x$, and a random
value $\varepsilon \left( y\right) $ that is specific to a particular
worker, holds for all firms of a given type $y$ and is known by the worker
at the time the matching occurs. This specification therefore contrasts with
the literature on search and matching that typically introduces a
match-specific shock that is revealed after the matching occurred.

In particular, we assume the value for a worker of type $x$ of working for a firm of
type $y$ at wage $w\left( x,y\right) $ is given by $\alpha
\left( x,y\right) +w\left( x,y\right)
+\sigma _{1}\varepsilon \left( y\right) $ where $\alpha
\left( x,y\right) +w\left( x,y\right) $ is deterministic, $\varepsilon
\left( y\right) $ is a worker-specific random process and $\sigma _{1}$ is a scaling factor. As in Dupuy and
Galichon (2014), we choose to model the random process $\varepsilon \left(
y\right) $ as a \emph{Gumbel random process}, introduced by Cosslett (1988)
and Dagsvik (1988), which is constructed as follows.\footnote{See appendix~\ref{app:ContinuousLogit} for more details about Gumbel random processes.} Assume
that, in a first step, workers form their demand by drawing a random pool of observable types
of firms, along with the corresponding utility shocks. We model the random
pool by a Poisson point process, so that its cardinality does not have to be
fixed and finite. More specifically, we assume that this Poisson process is
valued in $\mathcal{Y\times
\mathbb{R}
}$ with intensity $dye^{-\varepsilon }d\varepsilon $. The random pool sampled by a
worker is $\left\{ \left( y_{k},\varepsilon _{k}\right) ,k\in
\mathbb{N}
\right\} $, where $y_{k}$ is the type and $%
\varepsilon _{k}$ the corresponding utility shock. Define $%
\varepsilon \left( y\right) =\max_{k}\left\{ \varepsilon
_{k}:y_{k}=y\right\} $, with the convention that $\max \emptyset =-\infty $.

By construction, the problem of a utility-maximizing worker of type $x$ reads as
\begin{equation*}
\max_{{y}\in \mathcal{Y}} \left\{ \alpha \left( x,y\right) +w\left( x,y\right) +\sigma
_{1}\varepsilon \left( y\right) \right\} .
\end{equation*}

Note that workers' preferences only depend on their potential partner's type. Once the desired type has been determined, workers are indifferent between firms of that type.

By symmetry, we assume the value for a firm of type $y$ of hiring a worker of
type $x$ at wage $w\left( x,y\right) $ is given by $\gamma
\left( x,y\right) -w\left( x,y\right)
+\sigma _{2}\eta \left( x\right) $ where $\gamma
\left( x,y\right) -w\left( x,y\right) $ is deterministic, and $\eta
\left( x\right) $ is a firm-specific Gumbel random process.

Assuming that $\varepsilon \left( y\right)$ and $\eta \left( x\right)$ follow Gumbel random processes allows us to get a continuous logit framework. Indeed, Proposition~\ref{prop:continuousLogit} in the appendix, which was obtained by Cosslett (1988) and
Dagsvik (1988), shows that the density of demand for firms of type $y$ originating from workers of type $x$, is proportional to $\exp(\alpha(x,y)+w(x,y))$. This leads us assuming continuous logit demands from workers and firms, which we formalize in the following assumption.

\begin{condition}[Continuous logit demands]\label{Ass2}
All agents are price-takers and utility maximizers and given the wage
schedule $w\left( x,y\right) $, the conditional density demand for firms of type $y$ by workers of type $x$ is

\begin{equation}
\pi \left( y|x\right) =\frac{\exp \left( \frac{\alpha \left( x,y\right)
+w\left( x,y\right) }{\sigma _{1}}\right) }{\int_{\mathcal{Y}}\exp \left(
\frac{\alpha \left( x,y^{\prime }\right) +w\left( x,y^{\prime }\right) }{%
\sigma _{1}}\right) dy^{\prime }}.
\label{optimalityW}
\end{equation}%

Symmetrically, the conditional density demand for workers of type $x$ by firms of type $y$ is
\begin{equation}
\pi \left( x|y\right) =\frac{\exp \left( \frac{\gamma \left( x,y\right)
-w\left( x,y\right) }{\sigma _{2}}\right) }{\int_{\mathcal{X}}\exp \left(
\frac{\gamma \left( x^{\prime },y\right) -w\left( x^{\prime },y\right) }{%
\sigma _{2}}\right) dx^{\prime }}.
\label{optimalityF}
\end{equation}%

\end{condition}

Note that the conditional density demands rewrite as
\begin{equation}
\pi \left( y|x\right) = \exp \left( \frac{\alpha \left( x,y\right)
+w\left( x,y\right) -u\left( x\right) }{\sigma _{1}}\right)
\label{optimalityW2}
\end{equation}%
where $u\left( x\right) $ given by%
\begin{equation}
u\left( x\right) =\sigma _{1}\log \int_{\mathcal{Y}}\exp \left( \frac{\alpha
\left( x,y^{\prime }\right) +w\left( x,y^{\prime }\right) }{\sigma _{1}}%
\right) dy^{\prime }  \label{expru}
\end{equation}
interprets in the Gumbel framework as the expected indirect utility of a worker of type $x$ and symmetrically,
\begin{equation}
\pi \left( x|y\right) =\exp \left( \frac{\gamma \left( x,y\right)
-w\left( x,y\right) -v\left( y\right) }{\sigma _{2}}\right)
\label{optimalityF2}
\end{equation}%
where $v\left( y\right) $ given by%
\begin{equation}
v\left( y\right) =\sigma _{2}\log \int_{\mathcal{X}}\exp \left( \frac{\gamma
\left( x^{\prime },y\right) -w\left( x^{\prime },y\right) }{\sigma _{2}}%
\right) dx^{\prime }  \label{exprv}
\end{equation}
interprets as the expected indirect profits of a firm of type $y$.

In a second step, agents determine equilibrium by tatonnement over $w(x,y)$ using the demand functions defined in step 1. Note that since workers (resp. firms) are indifferent between firms (workers) of the same type, once the desired type has been determined, workers (resp. firms) can match with any of the firms (workers) of that type.

Note also that agents need not to observe other agents' idiosyncratic shocks $%
\varepsilon $ and $\eta $ to form their demand. Even if they had access to
that information, they would not use it.

\bigskip

\textbf{Equilibrium}. At equilibrium, the wage curve $w\left( x,y\right) $
is such that the density $\pi \left( x,y\right) $ of pairs $\left(
x,y\right) $ emanating from the workers' problem coincides with the density
of pairs $\left( x,y\right) $ emanating from the firms' problem, and hence must satisfy
\begin{equation}
\exp \left( \frac{\alpha \left( x,y\right) +w\left( x,y\right) -a\left(
x\right) }{\sigma _{1}}\right) =\pi \left( x,y\right) =\exp \left( \frac{%
\gamma \left( x,y\right) -w\left( x,y\right) -b\left( y\right) }{\sigma _{2}}%
\right),  \label{system_a_b}
\end{equation}


where
\begin{equation}
\QATOPD\{. {a\left( x\right) =u\left( x\right) -\sigma _{1}\log f\left(
x\right) }{b\left( y\right) =v\left( y\right) -\sigma _{2}\log g\left(
y\right) }.  \label{Exp_ab}
\end{equation}%

Substituting out $w\left( x,y\right) $ in system~(\ref{system_a_b}) yields%
\begin{equation}
\pi \left( x,y\right) =\exp \left( \frac{\phi \left( x,y\right) -a\left(
x\right) -b\left( y\right) }{\sigma }\right) ,  \label{OptMatch}
\end{equation}%
where $\sigma :=\sigma _{1}+\sigma _{2}$, while substituting out $\pi \left(
x,y\right) $ yields%
\begin{equation}
w\left( x,y\right) =\frac{\sigma _{1}}{\sigma }\left( \gamma \left(
x,y\right) -b\left( y\right) \right) +\frac{\sigma _{2}}{\sigma }\left(
a\left( x\right) -\alpha \left( x,y\right) \right) .  \label{EqWage}
\end{equation}

Note that $\sigma$ is the amount of heterogeneity
in the model and when the scaling factors of the random values $\sigma _{1}$
and $\sigma _{2}$ tend to zero, i.e. there is no heterogeneity in the model $%
\sigma \rightarrow 0$, the firm's problem and the worker's problem converge
to the deterministic maximization problems%
\begin{equation*}
u\left( x\right) =\max_{y\in \mathcal{Y}}\left\{ \alpha \left( x,y\right)
+w\left( x,y\right) \right\} \text{ and }v\left( y\right) =\max_{x\in
\mathcal{X}}\left\{ \gamma \left( x,y\right) -w\left( x,y\right) \right\}
\end{equation*}%
and the equilibrium problem consists in finding $w\left( x,y\right) $ and $%
\pi \left( x,y\right) $ which are compatible with optimality in these
problems. See section~\ref{rk:convergenceToZero} below.

Note also that in our model, equilibrium wages do not vary systematically with the idiosyncratic shocks of workers and firms.
This is because the non-wage valuation of a job at a firm of type $y$ by a worker of type $x$
is given as $\alpha \left( x,y\right) +\sigma _{1}\varepsilon \left(
y\right) $ and only depends on the observable characteristics of the firm, not
the firm's idiosyncratic shock (i.e. the process $\eta $). Symmetrically,
the productivity of a job at a firm of type $y$ when performed by a worker
of type $x$ is given as $\gamma \left( x,y\right) +\sigma _{2}\eta \left(
x\right) $, only depends on the observable characteristics of the worker,
not her idiosyncratic shock (i.e. the process $\varepsilon $). Hence, by the law of
one price, this implies that the wage of any worker of type $x$ must be the
same at every firm of the same type $y$.

We can now formally define an equilibrium outcome on this market.

\begin{definition}[Equilibrium outcome]
An equilibrium outcome $\left( \pi ,w\right) $ consists of an equilibrium
matching $\pi \left( x,y\right) $, and an equilibrium wage $w\left(
x,y\right) $ where there exist functions $a\left( x\right) $ and $b\left(
y\right) $ such that:

(i) matching $\pi $ is feasible: defined by~(\ref{OptMatch}) and $\pi \in \mathcal{M}\left( f,g\right) $, and

(ii) wage $w$ is defined by~(\ref{EqWage}).
\end{definition}

As a result, the equilibrium outcome problem consists of looking for
functions $a\left( x\right) $ and $b\left( y\right) $ that are solution to
the system
\begin{equation}
\left\{
\begin{array}{c}
\int_{\mathcal{Y}}\exp \left( \frac{\phi \left( x,y\right) -a\left( x\right)
-b\left( y\right) }{\sigma }\right) dy=f\left( x\right) \\
\int_{\mathcal{X}}\exp \left( \frac{\phi \left( x,y\right) -a\left( x\right)
-b\left( y\right) }{\sigma }\right) dx=g\left( y\right) .%
\end{array}%
\right.  \label{SBsyst}
\end{equation}

Two important remarks are in order.

\begin{remark}[Location normalization]
If $a(x)$ and $b(y)$ are solutions of system~(\ref{SBsyst}), so are $a(x)+t$
and $b(y)-t$. Using equation~(\ref{EqWage}), the equilibrium wages are $w\left( x,y\right) $ for
the former solution and $w\left( x,y\right) +t$ for the latter. The
nonuniqueness of the solution for system~(\ref{SBsyst}) requires a
normalization which is reflected by the arbitrary choice $a\left(
x_{0}\right) =0$ and a constant term $t$ in the equilibrium wages equation~(\ref%
{EqWage}). Uniqueness of such $\left( a,b\right) $ upon normalization $a\left( x_{0}\right) =0$
is proved in\ R\"{u}schendorf and Thomsen (1993), theorem 3.
\end{remark}

\begin{remark}[Continuous Mixed Logit demand]
It follows from formula~(\ref{optimalityW}) that the density of market
demand for firms of type $y$ is given by%
\begin{equation*}
\int_{\mathcal{X}}\frac{\exp \left( \frac{\alpha \left( x,y\right) +w\left(
x,y\right) }{\sigma _{1}}\right) }{\int_{\mathcal{Y}}\exp \left( \frac{%
\alpha \left( x,y^{\prime }\right) +w\left( x,y^{\prime }\right) }{\sigma
_{1}}\right) dy^{\prime }}dx
\end{equation*}%
which is a continuous Mixed Logit model. Likewise, the density of market
demand for workers of type $x$ has a similar expression. The equilibrium
wage $w\left( x,y\right) $ equates these quantities to the respective
densities of supply, $g(y)$ and $f(x)$ respectively.
\end{remark}

\section{Parametric estimation}

\subsection{Observations}

Assume that one has access to a random sample of the population of matches
of firms and workers. For each match, this sample contains information about
the worker's characteristics, her wage and the firm's characteristics. The
observations consist of $\left\{ \left( X_{i},Y_{i},W_{i}\right)
,i=1,...,n\right\} $, where $n$ is the number of observed matches, $i$
indexes an employer-employee match, $X_{i}$ and $Y_{i}$ \ are respectively
the vectors of employee's and employer's observable characteristics, which
are assumed to be sampled from a continuous distribution, and $W_{i}$ is a
noisy measure of the true unobserved transfer $w\left( X_{i},Y_{i}\right)$ assumed to be
such that%
\begin{equation}
W_{i}=w\left( X_{i},Y_{i}\right)+\epsilon _{i}  \label{eqWageObs}
\end{equation}%
where measurement error $\epsilon _{i}$ follows a $\mathcal{N}\left(
0,s^{2}\right) $ distribution and is independent of $\left(
X_{i},Y_{i}\right) $.

Note that depending on the nature of preferences, observed transfers $W_{i}$ can be a monotonic
transformation of observed wages. This flexibility allows one to consider
equation (\ref{eqWageObs}) as a hedonic wage regression using any known
monotonic transformation of wages, i.e. identity (to estimate the model in
levels), logarithm, power transformation etc.

Finally, note that, while we assume the analyst has access to data containing all
variables in $X$ and $Y$, in practice datasets only contain a
subset of these variables. In such a situation, the analyst faces
issues of unobserved heterogeneity that our current method does not account
for. However, while this is a current limitation of our approach, one can expect that existing methods dealing with unobserved heterogeneity would adapt to be incorporated into the framework. This is left for future research.

\subsection{Identification\label{rk:SeparateIdentification}}

In this section, we briefly discuss identification of the deterministic value
of amenities $\alpha $ and productivity $\gamma $. Not that $\alpha $ and $%
\gamma $ do not appear individually in the expression of the equilibrium
matches in equation~(\ref{OptMatch}), only the joint value of a match $\phi $
appears in this equation. However, $\alpha $ and $\gamma $ do appear
separately and with opposite signs in the formula of the equilibrium wages in
equation~(\ref{EqWage}).

This clearly indicates that when only matches are observed, one cannot
identify and hence estimate the deterministic value of amenities $\alpha $
separately from the deterministic value of productivity $\gamma $. In
contrast, if transfers are observed, one actually can identify and estimate
these objects separately.

It should be noted, however, that taking the values of $\sigma _{1}$ and $%
\sigma _{2}$ as known\footnote{The $\sigma$ parameters are not non-parametrically identified but they can be estimated using observed transfers once $\alpha(.,.)$ and $\gamma(.,.)$ have been parametrically specified.} and $\sigma _{1}+\sigma _{2}=1$ for notational
simplicity, equations~(\ref{optimalityW}) and~(\ref{optimalityF}) clearly
indicate that $\alpha \left( x,y\right) +w\left( x,y\right) $ is identified
up to a function $c\left( x\right) $ by $\sigma _{1}\ln \pi \left(
x,y\right) +c\left( x\right) $, and $\gamma \left( x,y\right) -w\left(
x,y\right) $ is identified up to a function $d\left( y\right) $ by $\sigma
_{2}\ln \pi \left( x,y\right) +d\left( y\right) $. It follows that $\alpha $
is identified up to fixed effects $c\left( x\right) $ by%
\begin{equation*}
\alpha \left( x,y\right) =\sigma _{1}\ln \pi \left( x,y\right) -w\left(
x,y\right) +c\left( x\right) ,
\end{equation*}%
while $\gamma $ is identified up to fixed effects $d\left( y\right) $ by%
\begin{equation*}
\gamma \left( x,y\right) =\sigma _{2}\ln \pi \left( x,y\right) +w\left(
x,y\right) +d\left( y\right) .
\end{equation*}

This result has been used in a nonparametric setting by Galichon and
Salani\'{e} (2021) and Salani\'{e} (2015). In this paper, we exploit it in
a parametric setting using basis functions of $x$ and $y$ (see Section %
\ref{par:param}). Indeed, since $\alpha $ is identified up to fixed effects $%
c\left( x\right) $, the parametrization of $\alpha $ can only include basis
functions depending on both $x$ and $y$ or on $y$ only but it cannot include
basis functions depending on $x$ only. By a similar reasoning, the
parametrization of $\gamma $ cannot include basis functions depending on $y$
only.

\subsection{Notation}

For the sake of readability and to avoid additional notational burden, we
propose the following change of notation. Replace $\alpha $ by $\sigma
\alpha $, $\gamma $ by $\sigma \gamma $, $\phi $ by $\sigma \phi $, $a$ by $%
\sigma a$, and $b$ by $\sigma b$, so that the equations of the model become%
\begin{equation}
\pi \left( x,y\right) =\exp \left( \phi \left( x,y\right) -a\left( x\right)
-b\left( y\right) \right) ,  \label{newLikelihood}
\end{equation}%
where $\left( a,b\right) $ is the unique solution to the system of equations
\begin{equation}
\left\{
\begin{array}{c}
\int_{\mathcal{Y}}\exp \left( \phi \left( x,y\right) -a\left( x\right)
-b\left( y\right) \right) dy=f\left( x\right) \\
\int_{\mathcal{X}}\exp \left( \phi \left( x,y\right) -a\left( x\right)
-b\left( y\right) \right) dx=g\left( y\right) ,%
\end{array}%
\right.  \label{newBernstein}
\end{equation}%
still normalized by $a\left( x_{0}\right) =0$, and the terms $a$ and $b$ are
related to $u$ and $v$ by
\begin{equation}
u\left( x\right) =\sigma a\left( x\right) +\sigma _{1}\log f\left( x\right)
+t\text{, and }v\left( y\right) =\sigma b\left( y\right) +\sigma _{2}\log
g\left( y\right) -t  \label{Welfare}
\end{equation}%
and the equilibrium transfer $w$ is given by
\begin{equation}
w\left( x,y\right) =\sigma _{1}\left( \gamma \left( x,y\right) -b\left(
y\right) \right) +\sigma _{2}\left( a\left( x\right) -\alpha \left(
x,y\right) \right) +t.  \label{newW}
\end{equation}

This change of notation is without loss of generality since from equation~(%
\ref{newW}) one can estimate parameters $\sigma _{1}$ and $\sigma _{2}$ and
hence $\sigma$ and therefore recover the initial values of $\alpha $ and $%
\gamma $. In the remainder of the paper, equations~(\ref{newLikelihood})--(%
\ref{newW}) will characterize the model to estimate.\footnote{%
The value of $\sigma$ is therefore not imposed but estimated. Note however that the non-negativity of $\sigma _{1}$ and $\sigma _{2}$ should be imposed as a
constraint, as in the application below.}

\subsection{Parametrization\label{par:param}}

Let $A$ and $\Gamma $ be two vectors of $\mathbb{R}^{K}$ parameterizing the
function of workers' systematic value of job amenities $\alpha $ and the
function of firms' systematic value of productivity $\gamma $, in a linear
way, so that
\begin{equation*}
\alpha (x,y;A)=\sum_{k=1}^{K}A_{k}\varphi _{k}(x,y)\text{, and }\gamma
(x,y;\Gamma )=\sum_{k=1}^{K}\Gamma _{k}\varphi _{k}(x,y),
\end{equation*}%
where the basis functions $\varphi _{k}$ are linearly independent, and may
include functions that depend on $x$ (respectively $y$) only. Note that by
definition, the function of the joint value of a match reads as%
\begin{equation}
\phi (x,y;\Phi )=\sum_{k=1}^{K}\Phi _{k}\varphi _{k}(x,y),  \label{phiparam}
\end{equation}%
where $\Phi _{k}=A_{k}+\Gamma _{k}$. Inspection of equation~(\ref{newW})
reveals that, given the parametric choice above, equilibrium matching and
transfers are parameterized by $\left( A,\Gamma ,\sigma _{1},\sigma
_{2},t\right) $. The model is hence fully parameterized by $\theta =\left(
A,\Gamma ,\sigma _{1},\sigma _{2},t,s^{2}\right) $, which we make explicit
by writing the predicted equilibrium transfer as $w(x,y;\theta )$.

\subsection{Estimation}

The main purpose of this exercise is to estimate the vector of parameters $%
\theta $. To this aim we adopt a maximum likelihood approach. It follows
from section~(\ref{rk:SeparateIdentification}) that the likelihood of
observing a pair $\left( x,y\right) $ only depends on $\Phi =A+\Gamma $, and
is given by%
\begin{equation*}
\pi (x,y;\Phi )=\exp \left( \phi \left( x,y;\Phi \right) -a(x;\Phi
)-b(y;\Phi )\right) ,
\end{equation*}%
where $a(x;\Phi )$ and $b(y;\Phi )$ are uniquely determined by system of
equations~(\ref{newBernstein}). Since, by assumption, measurement errors in
transfers are independent of $(X,Y)$, the log-likelihood of an observation $%
\left( x,y,w\right) $ at parameter $\theta $\ is therefore%
\begin{equation*}
\log L\left( x,y,w;\theta \right) =\log \pi \left( x,y;\Phi \right) -\frac{%
\left( w-w\left( x,y;\theta \right) \right) ^{2}}{2s^{2}}-\frac{1}{2}\log
s^{2},
\end{equation*}%
and hence, the log-likelihood of the sample reads as:%
\begin{equation}
\log L\left( \theta \right) =n\mathbb{E}_{\hat{\pi}}\left[ \phi \left(
X,Y;\Phi \right) -a\left( X;\Phi \right) -b\left( Y;\Phi \right) -\frac{%
\left( W-w\left( X,Y;\theta \right) \right) ^{2}}{2s^{2}}\right] -\frac{n}{2}%
\log s^{2}  \label{LLtheo}
\end{equation}%
where $\hat{\pi}\left( x,y\right) $ is the observed density of matches in
the data.

However, note that $a$, $b$ and $w$ that appear in~(\ref{LLtheo}) are
computed in the population; here, we only have access to a sample. So,
denoting $a_{i}$ and $b_{j}$ the sample analog of $a(x)$ and $b(x)$, we
compute the sample analog of system~(\ref{newBernstein})
\begin{equation}
\left\{
\begin{array}{c}
\sum_{j=1}^{n}\exp \left( \phi _{ij}\left( \Phi \right) -a_{i}-b_{j}\right)
=1/n,\forall i=1,...,n \\
\sum_{i=1}^{n}\exp \left( \phi _{ij}\left( \Phi \right) -a_{i}-b_{j}\right)
=1/n,\forall j=1,...,n%
\end{array}%
\right.  \label{SBsyst_sample}
\end{equation}%
with the added normalization $a_{1}=0$, which ensures uniqueness of the
solution.\footnote{%
An argument similar to theorem A.2 in Chernozhukov et al. (2017) would show
that the solution to the system~(\ref{SBsyst_sample}) converges uniformly to
$a$ and $b$ as computed in the population, i.e. solution of system~(\ref%
{newBernstein}).} (Note that since we have assumed that the population
distribution is continuous, each sampled observation occurs uniquely, hence
the right-hand side here is $1/n$; however, this could easily be extended to
a more general setting). We denote $\left( a_{i}\left( \Phi \right)
,b_{i}\left( \Phi \right) \right) $ this solution at $\Phi $. This allows us
to compute a sample estimate of the equilibrium transfer $w_{i}\left( \theta
\right) $ as%
\begin{equation}
w_{i}\left( \theta \right) :=\sigma _{1}\left( \gamma _{ii}\left( \Gamma
\right) -b_{i}\left( \Phi \right) \right) +\sigma _{2}\left( a_{i}\left(
\Phi \right) -\alpha _{ii}\left( A\right) \right) +t,  \label{eqWageSample}
\end{equation}%
where the notation $\alpha _{ij}\left( A\right) $ substitutes for $\alpha
\left( X_{i},Y_{j};A\right) $, and similarly for $\gamma _{ij}\left( \Gamma
\right) $.

We are thus able to give the expression of the log-likelihood of the sample
in our next result. Recall that $\theta =\left( A,\Gamma ,\sigma _{1},\sigma
_{2},t,s^{2}\right) $ and $\Phi =A+\Gamma $.

\begin{theorem}
\label{th:LogLikeSample}The log-likelihood of the sample is given by
\begin{equation}
\log \hat{L}\left( \theta \right) =\log \hat{L}_{1}\left( \theta \right)
+\log \hat{L}_{2}\left( \theta \right) ,  \label{eqlogLMEmp}
\end{equation}%
where%
\begin{equation}
\log \hat{L}_{1}\left( \theta \right) =\sum_{i=1}^{n}\left( \phi _{ii}\left(
\Phi \right) -a_{i}\left( \Phi \right) -b_{i}\left( \Phi \right) \right)
\label{ExpLL1}
\end{equation}%
and,%
\begin{equation}
\log \hat{L}_{2}\left( \theta \right) =-\sum_{i=1}^{n}\frac{\left(
W_{i}-w_{i}\left( \theta \right) \right) ^{2}}{2s^{2}}-\frac{n}{2}\log s^{2},
\label{ExpLL2}
\end{equation}%
where $\phi _{ij}\left( \Phi \right) :=\phi (X_{i},Y_{j};\Phi )$ is as in~(%
\ref{phiparam}), $a_{i}\left( \Phi \right) $ and $b_{i}\left( \Phi \right) $
are obtained as the solution of~(\ref{SBsyst_sample}), and where $%
w_{i}\left( \theta \right) $ is given by~(\ref{eqWageSample}).
\end{theorem}

\begin{proof}
Immediate given the discussion before the theorem.
\end{proof}

Theorem \ref{th:LogLikeSample} motivates the following remark.\footnote{Note that in most applications, the parameters of primary
interest are those governing workers' deterministic values of amenities and
firms' deterministic values of productivity, i.e. $A$ and $\Gamma $
respectively. The remaining parameters $\left( \sigma _{1},\sigma
_{2},t,s^{2}\right) $ are auxiliary. Our MLE estimator can be concentrated on the parameters of interest. These results are available in an online Appendix.}

\begin{remark}[Interpretation of the objective function]
Expression~(\ref{eqlogLMEmp}) has a straightforward interpretation. The term
$\log \hat{L}_{1}\left( \theta \right) $, whose expression is given in
equation~(\ref{ExpLL1}) comes from the observed matching patterns. It only
depends on $\theta $ through $\Phi =A+\Gamma $, and one has
\begin{equation*}
\frac{1}{n}\frac{\partial \log \hat{L}_{1}}{\partial \Phi _{k}}=\mathbb{E}_{%
\hat{\pi}}\left[ \varphi _{k}(X,Y)\right] -\mathbb{E}_{\pi ^{\Phi }}\left[
\varphi _{k}(X,Y)\right]
\end{equation*}%
where $\mathbb{E}_{\hat{\pi}}\ $is the sample average and $\mathbb{E}_{\pi
^{\Phi }}$ the expectation with respect to$\mathbb{\ }$%
\begin{equation*}
\pi _{ij}^{\Phi }:=\exp \left( \phi _{ij}\left( \Phi \right) -a_{i}\left(
\Phi \right) -b_{j}\left( \Phi \right) \right) .
\end{equation*}%
Hence, the contribution of the first term is to equate the predicted moments
of the matching distributions to their sample counterparts. The term $\log
\hat{L}_{2}\left( \theta \right) $, whose expression appears in equation~(%
\ref{ExpLL2}) tends to match the predicted transfers $w_{i}\left( \theta
\right) $ with the observed transfers $W_{i}$ in order to minimize the sum
of the square deviations $\left( W_{i}-w_{i}\left( \theta \right) \right)
^{2}$. Hence, the contribution of the second term is to equate the predicted
transfers with their sample counterparts. Of course, $s^{2}$ will determine
the relative weighting of those two terms in the joint optimization problem.
If $s^{2}$ is high, which means transfers are observed with a large amount
of noise, then the first term becomes predominant in the maximization
problem. In the limit $s^{2}\rightarrow +\infty $, the problem will boil
down to a two-stage problem, where the parameter $\Phi $ is estimated in the
first stage, and the rest of the parameters are estimated in the second
stage by Non-Linear Least Squares conditional on $A+\Gamma =\Phi $. In the
MLE\ procedure, $s^{2} $ is a parameter, and its value is determined by the
optimization procedure.
\end{remark}

\section{Application}

\subsection{Data}

We illustrate the usefulness of our method using an application to the
estimation of the value of job amenities related to risks of fatal injury.
This application requires access to a single cross-section of data
containing a representative sample of worker-job matches with information
about workers' characteristics (education, experience, gender etc.), their
(hourly) wage and a measure of fatality rates associated to their job. Many
surveys such as the CPS contain all required information but the fatal
injury data. As a result, following Thaler and Rosen (1978), a large strand
of the literature has compiled the required data by combining survey data
with data about fatality per type of jobs from alternative sources.

In this paper, we follow the recent work by Viscusi (2003, 2007 and 2013)
and construct measures of fatality rates by occupation-industry cells for
the period 2012-2016. Unfortunately, data on fatal injuries by occupation
within industries are not readily available. Instead, we rely on fatal
injury data by occupation (4-digits SOC) and by industry (4-digits NAICS)
provided by the U.S. Bureau of Labor Statistics (BLS) CFOI.\footnote{%
In the CFOI data, a fatal injury is an injury leading to death within one
year of the day of the accident. See Viscusi (2003) for more details about
the CFOI data and its use in the present context.} For each year in the
period 2012-2016, we create a matrix of fatal injuries by occupation$\times $%
industry by simply multiplying the marginal distribution by occupation and
industry hence assuming independence. To reduce measurement errors, the
4-digits occupational codes are aggregated into 25 major occupations and the
4-digits industry codes into 80 major industries.\footnote{%
We use crosswalks provided by the BLS to perform the aggregation to major
occupations and industry.} We then use the CPS March files\footnote{%
The BLS advises to use March files of the CPS for computations of total
employment.} for 2012-2016 and compute matrices of hours-adjusted employment
level by occupation and industry, combining person-weights, computed to this
effect by the census and the BLS, and hours worked per week. The two sets of
matrices are then merged allowing us to compute, for each year, fatality
rates for a given occupation-industry cell as the ratio of the number of
fatalities to total hours-weighted employment in that cell (see e.g.
Viscusi, 2013). To attenuate further measurement errors, for each
occupation-industry cell, we take the average fatality rate over time as our
measure of risk.

We obtained our working dataset by merging the 2017 March CPS data with our
measure of fatality rate by occupation-industry cells. This dataset
therefore contains information about our main variables of interest: hourly earnings, hours of work, gender,
years of schooling, age, ethnic group, marital status, whether one's job is in the public sector or not, and occupation-industry
fatality rates. We follow the literature (e.g. Viscusi, 2013) and keep only
full-time, non-agricultural, non-armed force workers\footnote{Assumption \ref{Ass2} requires all agents to be price-takers. This assumption is likely not to be met in the armed force industry whose sole employer is the US government. For this reason, we exclude armed force workers from the analysis. Note, however, that this exclusion is common in the hedonic wage regression literature, e.g. Viscusi (2013).} between 16 and 64 years
old for the remainder of the analysis.\footnote{%
As is standard when using March CPS wage data, see Katz and Murphy (1992)
for instance, the sample excludes individuals with hourly earnings below one
half of minimum wage and top coded earnings are imputed 1.45 times the top
code value.}

Table (\ref{TabDesc}) provides descriptive statistics of our working
dataset. The average fatality rate in our sample is about %
  3.44per $100,000$ which is close to the figure obtained
in Viscusi (2013) for the year 2008, i.e. 3.29. To further compare our
dataset with the literature, we run a hedonic (log)wage regression including
the traditional controls (gender, years of schooling, age, age squared,
ethnic group, marital status, union membership, public sector dummy,
regional dummies, Metropolitan dummy) and our measure of hours-weighted
fatality rates by occupation-industry. Using the estimate of the
compensating wage differential for risk, we obtain an estimate of the VSL of
\$9.7 million (\$2017). This figure falls in the range of estimates in the
literature using similar data, i.e. Viscusi (2013)'s estimate of \$8.4
million (\$2017) using the 2008 CPS data.\footnote{%
The risk coefficient in the log wage hedonic regression reported
in Viscusi (2013), i.e. 0.0024, is very close to the estimate obtained with
our data, i.e. 0.0027. Using 0.0024 instead of 0.0027 to calculate the VSL
with our data one would obtain \$8.6 million (\$2017).}

\subsection{Estimates}

We estimate the model using the maximum likelihood estimator
presented in this paper. Observed transfers are assumed to be the logarithm
of observed wages to be consistent with the hedonic regression literature
that typically uses log wage regressions. We standardize other continuous
variables to facilitate the comparison and interpretation (in terms of
standard deviation) of the respective coefficients.

Estimation requires to specify the basis functions used to parameterize the
values of job amenities $\alpha \left( x,y;A\right) $ and productivity $%
\gamma (x,y;\Gamma )$. We adopt a linear (in parameters) specification of the basis
functions and present estimates for the following specification:\footnote{%
The fit of this specification can be compared with that of
alternative (nesting/nested) specifications using likelihood ratio tests.
For instance, the test statistic obtained when comparing the chosen
specification with a richer specification, where both job amenities and productivity include interactions between workers' years of
schooling, experience and gender with jobs' risk and sector, is equal to
1.358. This statistic is not significantly different from 0 at conventional
levels. One concludes that the specification presented in
the paper should be preferred. Note however, that the estimates of the VSL are similar across specifications.}

\begin{equation*}
\alpha (x,y;A)=\sum_{l=1}^{2}A_{0,l}x^{(0)}y^{(l)}+A_{1,2}x^{(1)}y^{(2)},
\end{equation*}%
and
\begin{equation*}
\gamma (x,y;\Gamma )=\sum_{k=1}^{8}\Gamma
_{k,0}x^{(k)}y^{(0)}+\sum_{k=1}^{4}\sum_{l=1}^{2}\Gamma _{k,l}x^{(k)}y^{(l)}
\end{equation*}
where $x$ includes a constant (k=0), years of schooling (k=1), (potential%
\footnote{%
Age less years of schooling less 6.}) experience (k=2), experience squared
(k=8), a dummy variable for female (k=3), a dummy variable indicating
whether one is married or not (k=4), and 3 ethnic dummy variables (white,
black and asian, using others, incl. hispanic, as the reference group,
k=5,6,7), whereas $y$ includes a constant (l=0), our measure of fatality
rates (l=1) and a dummy variable indicating the public sector (l=2).

Hence, our specification of job amenities includes the main effects of fatality
rates and public sector as well as an interaction between a workers' years
of schooling and jobs' sector. Our specification of perceived productivity includes
the main effects of years of schooling, experience (squared), marital status
and ethnic groups as well as interactions between workers' year of
schooling, experience and gender with jobs' fatality rates and sector.

Estimates are presented in table (\ref{TableMainEffects_QE10}). Note first that the model fits quite well the wage data with an $R^{2}$ of $%
0.235$ which compares to that obtained for the standard hedonic
wage regression, i.e. $R^{2}=0.255$.

Second, estimates of the value of perceived productivity show expected results.%
\footnote{%
Unless stated otherwise, the significance level is 1\%.} The value of
productivity increases with years of schooling ($0.057$), although the
estimate is not significant,\footnote{%
We have also estimated the model including years of schooling squared.
However, comparing the two specifications, the log-likelihood ratio test statistic is 3.360 and not significantly different
from 0 at conventional levels. Our chosen specification should be preferred.} and the experience-productivity
gradient is positif ($0.084$) but decreasing, as indicated by the negative
coefficient for experience squared ($-0.051$). These human capital effects,
however, vary significantly across jobs: the years-of-schooling-productivity gradient is absent in risky jobs ($%
-0.002 = 0.057-0.059$) but greater for public sector jobs ($0.895 =
0.057+0.838$), whereas the experience-productivity gradient is higher in risky jobs ($0.158 =
0.084+0.074$).

Third, our estimates for perceived productivity show negative coefficients for female
and black workers ($-0.404$ and $-0.108$ respectively) and a positive
coefficient for white workers ($0.046$). These coefficients should be interpreted
with care as they indeed reflect employers' perceived productivity of the
underlying types of workers, revealing discrimination effects.\footnote{We refer herewith to Becker (1971). The parameter $\gamma$ reflecting both the true productivity of workers and employers' taste discrimination parameter.} Our results are in line with the
large literature showing discriminating wage differentials across gender and race. Interestingly, the gender perceived
productivity gap varies significantly across jobs unlike the racial one: a
one standard deviation increase in the probability of fatal injury more than
triples the gender perceived productivity gap.

Fourth, regarding the value of job amenities, results show that the value of
public sector jobs increases significantly with years of schooling: a one
standard deviation in years of schooling generates a 0.081 log-points
increase in the value of jobs in the public sector.

Finally, our main result shows that US workers' utility drops by %
 0.023log-points per hour of work as the
probability of fatal injury on the job increases by one standard deviation
(i.e.  13.05per $100,000$). We can use this coefficient to
compute the VSL from the formula
\begin{equation*}
VSL\left( x,y\right) =-\frac{\partial \alpha \left( x,y\right) }{\partial
y^{(1)}}\overline{z},
\end{equation*}
where $\overline{z}$ are the average earnings in the sample.\footnote{%
To derive this formula, remember that the systematic utility of a worker of
type $x$ working in job of type $y$ and receiving a transfer $w\left(
x,y\right)$ is given as $U\left( x,y\right) =\alpha\left( x,y\right) +w\left(x,y\right)$.
Since transfers are specified in log wages, i.e. $w\left( x,y\right) =\ln
z\left( x,y\right) $ where $z\left( x,y\right) $ are the equilibrium wages,
we can then compute a worker of type $x$'s trade-off between earnings and
risk as
\begin{equation*}
\frac{\partial U\left( x,y\right) }{\partial y^{(1)}} := \frac{\partial
\alpha \left( x,y\right) }{\partial y^{(1)}}+\frac{\frac{\partial z\left(
x,y\right) }{\partial y^{(1)}}}{z\left( x,y\right) } = 0.
\end{equation*}
Rearranging this equation, one can express the differential wage increase
required to compensate the differential drop in job amenity due to an one
unit increase in the risk of fatal injury as in the text %
using $\overline{z}$ to replace $z\left( x,y\right) $.}

Using $\frac{\partial \alpha \left( x,y\right) }{\partial y^{(1)}}=A_{1,1}=%
 0.023$ and the appropriate units, one obtains a VSL
of \$   6.3million (\$2017).\footnote{%
In our preferred specification, the VSL does not vary with the
type of workers nor with the fatality risk. Note that this method
excludes productivity effects of fatality risk. It only reflects the
valuation of life from the perspective of workers.} This value lies in the
range of estimates found in the literature using similar data, i.e. Viscusi,
(2013). Nevertheless, it is about \$3 million lower than the estimate
obtained using the classical hedonic wage regression. Though the difference
is not statistically significant, it suggests that not accounting for the
sorting of workers into jobs, as in the hedonic regression, may lead to an
overestimation of the true VSL in our data.

To see this, note that, in our model, the differential value of job
amenities with respect to fatality risk is identified as
\begin{equation*}
\frac{\partial \alpha \left( x,y\right) }{\partial y^{(1)}}=\sigma _{1}\frac{%
\partial \ln \pi \left( y|x\right) }{\partial y^{(1)}}-\frac{\partial
w\left( x,y\right) }{\partial y^{(1)}}.
\end{equation*}

In contrast, the hedonic wage regression literature identifies this
differential value using the coefficients of an (log) earnings regression as
\begin{equation*}
\frac{\partial \alpha^{h}\left( x,y\right) }{\partial y^{(1)}} = -\frac{%
\partial w\left( x,y\right) }{\partial y^{(1)}} = \frac{\partial \alpha \left( x,y\right) }{\partial y^{(1)}}-\sigma _{1}%
\frac{\partial \ln \pi \left( y|x\right) }{\partial y^{(1)}}.
\end{equation*}


As a result, the VSL as measured in the hedonic regression literature reads
as
\begin{equation*}
VSL^{h}\left( x,y\right) :=-\frac{\partial \alpha^{h}\left( x,y\right) }{%
\partial y^{(1)}}\overline{z} = VSL\left( x,y\right) -\sigma _{1}\frac{\partial
\ln \pi \left( y|x\right) }{\partial y^{(1)}}\overline{z},
\end{equation*}%
once substituting $\frac{\partial \alpha^{h}\left( x,y\right) }{%
\partial y^{(1)}}$ by its expression in terms of $\alpha$ and $\pi$.

Since average wages are positive, $\overline{z}>0$, and $\sigma_{1} > 0$, it
follows that, compared to our method, estimates of VSL from hedonic
wage regressions tend to be larger (lower) when, in equilibrium, conditional
on workers' type, workers sort into safe (resp. risky) jobs, i.e. when $%
\frac{\partial \ln \pi \left( y|x\right) }{\partial y^{(1)}}<0$ (resp. $>0$).

Not only our structural approach allows one to explicitly take into account
the sorting of workers to jobs when estimating the value of job amenities,
it also allows to compute counterfactual equilibria. In particular, one
could use the estimates of job amenities and productivity obtained above to
compute the impact of a government intervention aiming at reducing fatality
risk at work. For instance, consider the Site-Specific Targeting (SST)
inspection plan proposed by the Occupational Safety and Health
Administration (OSHA) in the US. If effective, this program would decrease
fatality rates in the most risky jobs. For the sake of an example, suppose
that the program ends up decreasing the fatality risk of all jobs whose
fatality risk is at least 1 standard deviation above the mean (i.e. $\geq
16.5=3.4+1\times 13.1$) down to $16.5$\ fatalities per 100,000 workers per
year. As a result of this intervention, the distribution of types of jobs
would change causing the equilibrium matching and wages to change too. We
use our model to compute the equilibrium before (observed) and after the
intervention and then compare the matching and distribution of wages. We
find that, as a result of this intervention, about 3.1\% of the workers
would change jobs, the mean wage would drop by 3.9\% and wage inequality, as
measured by the Gini coefficient, would drop by 3.6\%.

\section{Discussion and conclusion}

We conclude by addressing a few methodological remarks before summarizing our main results.


\subsection{Job seekers and vacancies}
A natural extension of the model is to allow for workers to be unemployed
and jobs to be vacant. To do so, one needs first to allow for the total masses of workers and
firms to be different, herewith relaxing that part of Assumption~\ref{Ass1}. Second, one needs to extend the
definition of utilities for matched workers and firms to unemployed workers
and vacant jobs, introducing reservation utilities. The reservation utility of a worker (firm) of type $x$ ($y$%
) may be decomposed into a systematic part $\alpha (x,\emptyset )$ ($\gamma
(\emptyset ,y)$) and a random value $\varepsilon \left( \emptyset \right) $ (%
$\eta \left( \emptyset \right) $) following a Gumbel type I distribution.
Assumption~\ref{Ass2} should then simply be modified by replacing the choice
sets of workers $\mathcal{Y}$ and firms $\mathcal{X}$ by $\mathcal{Y}%
\cup \left\{ \emptyset \right\} $ and $\mathcal{X\cup }\left\{ \emptyset
\right\} $ respectively and adopting the convention that $w\left(
x,\emptyset \right) =w\left( \emptyset ,y\right)=0 $ for all types of workers
and firms.

However, note that the conditional probabilities in (\ref{optimalityW}) and (\ref{optimalityF}) have a logit structure and
hence satisfy the Independence of Irrelevant Alternative property. As a
consequence, $\pi \left( y|x\right) $ in~(\ref{optimalityW}) is also the
density of probability of choosing a firm of type $y$ for a worker of type $x
$ conditional on participation. As shown in appendix D in Dupuy and Galichon
(2014), this implies that in a market with unemployed workers and vacant jobs and reservation utilities $\alpha (x,\emptyset )+\varepsilon \left(
\emptyset \right) $ and $\gamma (\emptyset ,y)+\eta \left( \emptyset \right),$ at equilibrium, the probability density $\pi
\left( x,y\right) $ of occurrence of a $\left( x,y\right) $ pair among matched pairs (i.e. considering only active workers and filled jobs) is the same as the probability density $\pi \left( x,y\right) $ of occurrence of a $\left( x,y\right) $ pair in a market with no outside options and where the masses of workers and firms are the same as the masses of active workers and filled jobs in the former market.

\subsection{Related assignment models}\label{rk:convergenceToZero}
When $\sigma \rightarrow 0$, the model converges
to the classical model of Monge-Kantorovich, which is a continuous extension
of the Becker-Shapley-Shubik model. Indeed, when $\sigma _{1}$ and $\sigma
_{2}$ tend to zero, the scaling coefficients of the random value of job
amenities and productivity $\varepsilon $ and $\eta $, tend to zero, then
the model becomes nonstochastic. Intuitively, when $\sigma _{1}\rightarrow 0$%
, the worker's expected indirect utility $u\left( x\right) $ tends to the
deterministic indirect utility $\max_{y}\left\{ \alpha \left( x,y\right)
+w\left( x,y\right) \right\} $, and it follows from~(\ref{optimalityW}) that
the conditional choice distribution $\pi \left( y|x\right) $ becomes
concentrated around the optimal firm's type $y$ such that $u\left( x\right)
=\alpha \left( x,y\right) +w\left( x,y\right) $. Similarly, when $\sigma
_{2}\rightarrow 0$, a firm of type $y$ expected indirect profits $v\left(
y\right) $ tends to the deterministic indirect profits $\max_{x}\left\{
\gamma \left( x,y\right) -w\left( x,y\right) \right\} $, and $\pi \left(
x|y\right) $ becomes concentrated around the optimal worker's type $x$ such
that $v\left( y\right) =\gamma \left( x,y\right) -w\left( x,y\right) $.
Combining these two results, $\pi \left( x,y\right) $ becomes concentrated
around the set of pair $\left( x,y\right) $ such that $u\left( x\right)
+v\left( y\right) =\phi \left( x,y\right) $, hence, in the limit when $%
\sigma _{1}$ and $\sigma _{2}$ tend to zero, we have%
\begin{equation*}
\left\{
\begin{tabular}{rccl}
$\pi $ & $\in $ & $\mathcal{M}\left( f,g\right) $ &  \\
$u\left( x\right) +v\left( y\right) $ & $\geq $ & $\phi \left( x,y\right) $
& $\forall x\in \mathcal{X},y\in \mathcal{Y}$ \\
$u\left( x\right) +v\left( y\right) $ & $=$ & $\phi \left( x,y\right) $ & $%
\pi -a.s.$%
\end{tabular}%
\right.
\end{equation*}%
These are the classical stability conditions in the Monge-Kantorovich
problem (see Villani, 2003 and 2009), whose variants have been applied in
economics by Becker (1973), Shapley and Shubik (1963), Gretsky, Ostroy, Zame
(1992).
A particular example is Sattinger's workhorse model extensively used in the labor economics literature (see Sattinger, 1979 and
1993). This model indeed corresponds to a matching market with no unobserved
heterogeneity ($\sigma \rightarrow 0$), unidimensional observed types ($d_{x}=d_{y}=1$), in which workers only
care about their compensation ($\alpha =0$) and where the firm's value of
productivity is smooth and supermodular (i.e. $\partial ^{2}\gamma \left(
x,y\right) /\partial x\partial y$ exists and is positive). Under these
restrictions, both the worker's and firm's problems become deterministic,
and the conditional distribution $\pi (y|x)$ in this case is concentrated at
one point $y=T(x)$, where $T(x)$ is the only assignment of workers to firms
which is nondecreasing. The equilibrium wage $w$ only depends on $x$ and
satisfies the differential wage equation $w^{\prime }\left( x\right) =\frac{\partial \gamma \left( x,T\left( x\right)
\right) }{\partial x}$ and an explicit formula for the equilibrium wage is obtained by integration.

\subsection{Conclusion}

Over the last decade, a great deal of efforts has been made to bring
matching models to data. In the transferable utility class of models,
following Choo and Siow's seminal contribution, various extensions have been
proposed to enrich the empirical methodology. These extensions were so far
limited to the case when transfers are not observed. However, the
observation of transfers allows to widen the scope of identified objects in
this class of models, and in particular allows the analyst to separately
identify the (pre-transfer) values of a match for each partner. Our paper
proposes an intuitive and tractable maximum likelihood approach to
structurally estimate these values of a match for each partner using data
about matches and transfers from a single market.

We illustrate the usefulness of our methodology to the estimation of
compensating wage differentials for the risk of fatal injury on the job.
Using the 2017 March CPS data together with CFOI data on fatal injury per
occupation and industry, our estimate of the value of job amenities related
to risk translates into a Value of Statistical Life of \$%
   6.3million (\$2017). This estimate is \$3 million lower
(though not significantly) than the one obtained by applying a classical
hedonic regression technique on our data. Since the hedonic approach can be
seen as an extreme version of our method where all the weight in the
likelihood function is put on fitting transfers (wages) and none on fitting
matching patterns, this suggests that not accounting explicitly for the
sorting of workers to jobs can lead to biases in the estimation of the value
of statistical life.

\newpage

\section*{Tables\label{app:Tables}}

\begin{petit}
\begin{table}[tbph]
\caption{Descriptive statistics of workers' and firms' attributes and hourly
wages (in 2017 dollars).}
\label{TabDesc}\centering\medskip 
\begin{tabular}{l|rrrr}
\hline
  & Mean & Std & Min & Max \\ \hline 
Workers &  & & &  \\  
Years of Schooling (in years)  &  13.35 &   2.24  &   1.00 &  21.00 \\ 
Experience (in years)          &  20.67 &  12.97  &   0.00 &  51.00 \\ 
Female                         &   0.52 &   0.50  &   0.00 &   1.00 \\ 
Married                        &   0.50 &   0.50  &   0.00 &   1.00 \\ 
White                          &   0.63 &   0.48  &   0.00 &   1.00 \\ 
Black                          &   0.12 &   0.32  &   0.00 &   1.00 \\ 
Asian                          &   0.06 &   0.24  &   0.00 &   1.00 \\ 
Wage (hourly)                  &  17.95 &   9.02  &   3.75 &  70.00 \\ 
\hline Firms &  & & & \\  
Risk (per 100,000)            &     3.44 &  13.05  &   0.00 & 345.70 \\ 
Public                        &     0.12 &   0.33  &   0.00 &   1.00 \\ 
\hline N &   3454 &   &  &  \\ 
\hline\hline
 &  &  &  &
\end{tabular}%
\par
\begin{note}
Note: For measurement purposes, a fatal injury is an injury leading to death
within one year of the day of the accident.
\end{note}
\end{table}
\end{petit}

\newpage

\begin{petit}
\begin{table}[tbph]
\caption{Effect of firms' and workers' attributes on job amenities and
(perceived) productivity (in 2017 dollars), specification 3.}
\label{TableMainEffects_QE10}\centering\medskip 
\begin{tabular}{l|ccc}
\hline
 Job Amenities (Alpha) &Main effects &          Risk (in 100,000) &     Public                   \\ 
Main effects          &  & -0.023 & -0.062 \\ 
 &  & ( 0.009) & ( 0.027) \\ 
YoS (in years)        &        &        &  0.081 \\ 
 &        &        & ( 0.031) \\ 
\hline  Productivity (Gamma) & Main effects & Risk (in 100,000) & Public  \\  
YoS (in years)                &  0.057 & -0.059 &  0.838  \\ 
 & ( 0.035) & ( 0.020) & ( 0.099) \\ 
Experience (in years)         &  0.084 &  0.074 &  0.096  \\ 
 & ( 0.025) & ( 0.029) & ( 0.104) \\ 
Female                        & -0.404 & -2.388 &  0.548  \\ 
 & ( 0.061) & ( 0.238) & ( 0.212) \\ 
Married                       &  0.050 &  &  \\ 
 & ( 0.020) &  &  \\ 
White                         &  0.046 &  &  \\ 
 & ( 0.021) &  &  \\ 
Black                         & -0.108 &  &  \\ 
 & ( 0.039) &  &  \\ 
Asian                         &  0.069 &  &  \\ 
 & ( 0.035) &  &  \\ 
Experience squared (in years) & -0.051 &  &  \\ 
 & ( 0.016) &  &  \\ 
\hline Salary constant &  2.981 &  &  \\ 
 & ( 0.373) &  &  \\ 
Sigma 1 &  0.046 &  &  \\ 
Sigma 2 &  2.233 &  &  \\ 
\hline R-square &  0.235 &  &  \\ 
\hline\hline 
 &  &  &
\end{tabular}%
\par
\begin{note}
Notes: This table reports the estimates of the main effects of workers'
characteristics on (perceived) productivity and firms' characteristics on
job amenities as well as the interaction of
workers' characteristics and firms' characteristics on (perceived)
productivity and job amenities. All effects are measured in dollars (per hour of work). All non dummy
covariates are standardized to have a standard deviation of 1. Standard
errors, calculated from the Hessian of the likelihood, are in parentheses.
\end{note}
\end{table}
\end{petit}

\appendix

\section*{Appendix}

\section{\label{app:ContinuousLogit}The continuous logit framework}

Recall that the value for a worker $x$ of the job amenities at firm $y$ is
given by $U\left( x,y\right) +\sigma _{1}\varepsilon \left( y\right) $ where
$U\left( x,y\right) =\alpha \left( x,y\right) +w\left( x,y\right) $ is
deterministic, and $\varepsilon \left( y\right) $ is a worker-specific
random process. As in Dupuy and Galichon (2014), we choose to model the
random process $\varepsilon \left( y\right) $ as a \emph{Gumbel random
process}, introduced by Cosslett (1988) and Dagsvik (1988).

Assume that workers form their demand
by drawing a random pool of observable types of firms, along with the corresponding
utility shocks. We call this pool a worker's \textquotedblleft random pool of prospects.\textquotedblright

Let $k\in
\mathbb{N}
$ index firms in a worker's pool of prospects and $\left\{ \left( y_{k},\varepsilon
_{k}\right) ,k\in
\mathbb{N}
\right\} $ be the points of a Poisson process on $\mathcal{Y\times
\mathbb{R}
}$\ with intensity $dye^{-\varepsilon }d\varepsilon $. A worker of type $x$
therefore chooses a firm's of type $y$ by looking at her pool of prospects and solving the utility maximization program
\begin{equation*}
\tilde{U}=\max_{y\in \mathcal{Y}}\left\{ U\left( x,y\right) +\sigma
_{1}\varepsilon \left( y\right) \right\} =\max_{k\in
\mathbb{N}
}\left\{ U\left( x,y_{k}\right) +\sigma _{1}\varepsilon _{k}\right\} ,
\end{equation*}%
where $\tilde{U}$ denotes the worker's (random) indirect utility. The
worker's program induces conditional density of choice probability of firm's
type given worker's type, which is expressed as follows:

\begin{proposition}
\label{prop:continuousLogit}The conditional density of probability of
choosing a firm of type $y$ for a worker of type $x$ is given by%
\begin{equation*}
\pi \left( y|x\right) =\frac{\exp \left( \frac{U\left( x,y\right) }{\sigma
_{1}}\right) }{\int_{\mathcal{Y}}\exp \left( \frac{U\left( x,y^{\prime
}\right) }{\sigma _{1}}\right) dy^{\prime }}
\end{equation*}%
while the expected indirect utility of a worker of type $x$, denoted $u\left( x\right)
=\mathbb{E}\left[ \tilde{U}|x\right] $, is expressed as%
\begin{equation*}
u\left( x\right) =\sigma _{1}\log \int_{\mathcal{Y}}\exp \left( \frac{%
U\left( x,y^{\prime }\right) }{\sigma _{1}}\right) dy^{\prime }.
\end{equation*}
\end{proposition}

This result was obtained by Cosslett (1988) and Dagsvik (1988). The
intuition of the result is that the c.d.f. of the random utility $\tilde{U}$
conditional on $X=x$ is given by $F_{\tilde{U}|X=x}\left( z|x\right) =\Pr
\left( \tilde{U}\leq z|X=x\right) $, which is the probability that the
process $\left( y_{k},\varepsilon _{k}\right) $ does not intersect the set $%
\left\{ \left( y,e\right) :U\left( x,y\right) +\sigma _{1}e>z\right\} $.
Hence, the log probability of the event $\tilde{U}\leq z$ is minus the
integral of the intensity of the Poisson process over this set, that is
\begin{eqnarray*}
\log \Pr \left( \tilde{U}\leq z|X=x\right) &=&-\int_{\mathcal{Y}}\int_{%
\mathcal{%
\mathbb{R}
}}1\left\{ U\left( x,y\right) +\sigma _{1}e>z\right\} e^{-\varepsilon }dedy
\\
&=&-\exp \left( -z+\log \int_{\mathcal{Y}}\exp \left( \frac{U\left(
x,y\right) }{\sigma _{1}}\right) dy\right) ,
\end{eqnarray*}%
which is the c.d.f. of a Gumbel distribution with location parameter $\log
\int_{\mathcal{Y}}\exp \left( U(x,y)\right) dy$, and scale parameter $\sigma
_{1}$.

\section{Proofs\label{app:proofs} and additional results}

Let $Da$ and $Db$ be the two $n\times K$ matrices of respective terms $%
\partial a_{i}\left( \Phi \right) /\partial \Phi _{k}$ and $\partial
b_{j}\left( \Phi \right) /\partial \Phi _{k}$ respectively. Let $\Pi $ be
the matrix of terms $\pi _{ij}^{\Phi }=\exp \left( \phi _{ij}\left( \Phi
\right) -a_{i}\left( \Phi \right) -b_{j}\left( \Phi \right) \right) $, and
let $\tilde{\Pi}$ be the same matrix where the entries on the first row have
been replaced by zeroes. Let $E$ be the $n\times K$ matrix whose terms $%
E_{ik}$ are such that $E_{1k}=0$ for all $k$, and $E_{ik}=\sum_{j=1}^{n}\pi
_{ij}^{\Phi }\varphi _{k}\left( x_{i},y_{j}\right) $ for $i\geq 2$ and all $%
k $. Let $F$ be the $n\times K$ matrix of terms such that $%
F_{jk}=\sum_{i=1}^{n}\pi _{ij}^{\Phi }\varphi _{k}\left( x_{i},y_{j}\right) $%
.

\begin{lemma}
\label{lem:GradPotentials}The derivatives of the $a_{i}$'s and the $b_{i}$'s
with respect to the $\Phi _{k}$'s are given by matrices $Da$ and $Db$ such
that%
\begin{equation}
\binom{Da}{Db}=%
\begin{pmatrix}
I & \tilde{\Pi} \\
\Pi ^{\top } & I%
\end{pmatrix}%
^{-1}\binom{E}{F}.  \label{expressionDaDb}
\end{equation}
\end{lemma}

\begin{proof}[Proof of lemma~\protect\ref{lem:GradPotentials}]
Recall that
\begin{equation*}
\left( Da\right) _{ik}:=\frac{\partial a_{i}\left( \Phi \right) }{\partial
\Phi _{k}}\text{ and }\left( Db\right) _{jk}:=\frac{\partial b_{j}\left(
\Phi \right) }{\partial \Phi _{k}}
\end{equation*}%
for $1\leq i\leq n$ and $1\leq k\leq K$. Note that the system in equation~(%
\ref{newBernstein}) is normalized such that $a_{1}\left( \Phi \right) =0$,
one has that $\partial a_{1}\left( \Phi \right) /\partial \Phi _{k}=0$ for
all $k$. Differentiation yields%
\begin{eqnarray*}
Da_{1k} &=&0 \\
Da_{1k}+\sum_{j=1}^{n}\pi _{ij}^{\Phi }Db_{jk} &=&E_{ik},~i\in \left\{
2,...,n\right\} \\
\sum_{i=1}^{n}\pi _{ij}^{\Phi }Da_{ik}+Db_{jk} &=&F_{jk},~j\in \left\{
1,...,n\right\} ,
\end{eqnarray*}%
where $\pi _{ij}^{\Phi }=\exp \left( \phi _{ij}\left( \Phi \right)
-a_{i}\left( \Phi \right) -b_{j}\left( \Phi \right) \right) $. Recall that
under the linear parameterization we have adopted in section~\ref{par:param}%
, $\partial \phi _{ij}\left( \Phi \right) /\partial \Phi _{k}=\varphi
_{k}(x_{i},y_{j})$ and let%
\begin{eqnarray*}
E_{1k} &=&0\text{, }E_{ik}=\sum_{j=1}^{n}\pi _{ij}^{\Phi }\varphi
_{k}(x_{i},y_{j})\text{ for }i\geq 2,\text{ and} \\
F_{jk} &=&\sum_{i=1}^{n}\pi _{ij}^{\Phi }\varphi _{k}(x_{i},y_{j})\text{ for
all }j\text{,}
\end{eqnarray*}%
this system rewrites%
\begin{equation}
\begin{pmatrix}
I & \widetilde{\Pi } \\
\Pi ^{\top } & I%
\end{pmatrix}%
\binom{Da}{Db}=\binom{E}{F}  \label{linrelat}
\end{equation}%
where block $\tilde{\Pi}$ is the $n\times n$ matrix of term $\widetilde{\pi }%
_{ij}^{\Phi }$ so that$\ \widetilde{\pi }_{1j}^{\Phi }=0$ for all $j\in
\left\{ 1,...,n\right\} $ and $\widetilde{\pi }_{ij}^{\Phi }=\pi _{ij}^{\Phi
}$ for $i\geq 2$ and all $j\in \left\{ 1,...,n\right\} $, and block $\Pi $
is \ the $n\times n$ matrix of term $\pi _{ij}^{\Phi }$. It is easily
checked that the matrix on the left hand-side of~(\ref{linrelat}) is
invertible. One therefore obtains $Da$ and $Db$ as%
\begin{equation*}
\binom{Da}{Db}=%
\begin{pmatrix}
I & \widetilde{\Pi } \\
\Pi ^{\top } & I%
\end{pmatrix}%
^{-1}\binom{E}{F}.
\end{equation*}
\end{proof}

Recall $\theta =\left( A,\Gamma ,\sigma _{1},\sigma
_{2},t,s^{2}\right) $ and $\Phi =A+\Gamma $.

\begin{theorem}
\label{thm:LLGradient}(i) The partial derivatives of $\log \hat{L}_{1}\left(
\theta \right) $ with respect to $A_{k}$ and $\Gamma _{k}$ are given by
\begin{equation*}
\frac{\partial \log \hat{L}_{1}\left( \theta \right) }{\partial A_{k}}=\frac{%
\partial \log \hat{L}_{1}\left( \theta \right) }{\partial \Gamma _{k}}%
=\sum_{i=1}^{n}\varphi _{k}(x_{i},y_{i})-n\sum_{i,j=1}^{n}\pi _{ij}^{\Phi
}\varphi _{k}(x_{i},y_{j})
\end{equation*}%
and the partial derivatives of $\log \hat{L}_{1}\left( \theta \right) $ with
respect to all the other parameters is zero.

(ii) The partial derivatives of $\log \hat{L}_{2}\left( \theta \right) $
with respect to any parameter entry $\theta _{k}$ other than $s$ is given by%
\begin{equation*}
\frac{\partial \log \hat{L}_{2}\left( \theta \right) }{\partial \theta _{k}}%
=s^{-2}\sum_{i=1}^{n}\left( W_{i}-w_{i}\left( \theta \right) \right) \frac{%
\partial w_{i}\left( \theta \right) }{\partial \theta _{k}}
\end{equation*}

(iii) The partial derivative of $\log \hat{L}_{2}\left( \theta \right) $
with respect to $s^{2}$ is given by%
\begin{equation*}
\frac{\partial \log \hat{L}_{2}\left( \theta \right) }{\partial s^{2}}%
=\sum_{i=1}^{n}\frac{\left( W_{i}-w_{i}\left( \theta \right) \right) ^{2}}{%
2s^{4}}-\frac{n}{2s^{2}}
\end{equation*}

(iv) The partial derivative of $w_{i}\left( \theta \right) $ with respect to
$t$ is one, its derivative with respect to $\sigma _{1}$ is $\gamma
_{ii}\left( \Gamma \right) -b_{i}\left( \Phi \right) $, its derivative with
respect to $\sigma _{2}$ is $a_{i}\left( \Phi \right) -\alpha _{ii}\left(
\Gamma \right) $. The partial derivative of $w_{i}\left( \theta \right) $
with respect to $A_{k}$ and $\Gamma _{k}$ are given by%
\begin{eqnarray*}
\frac{\partial w_{i}\left( \theta \right) }{\partial A_{k}} &=&\sigma
_{2}\left( \frac{\partial a_{i}\left( \Phi \right) }{\partial \Phi _{k}}%
-\varphi _{k}\left( x_{i},y_{i}\right) \right) -\sigma _{1}\frac{\partial
b_{i}\left( \Phi \right) }{\partial \Phi _{k}} \\
\frac{\partial w_{i}\left( \theta \right) }{\partial \Gamma _{k}} &=&\sigma
_{1}\left( \varphi _{k}\left( x_{i},y_{i}\right) -\frac{\partial b_{i}\left(
\Phi \right) }{\partial \Phi _{k}}\right) +\sigma _{2}\frac{\partial
a_{i}\left( \Phi \right) }{\partial \Phi _{k}}.
\end{eqnarray*}

(v) The partial derivatives $\partial a_{i}\left( \Phi \right) /\partial
\Phi _{k}$ and $\partial b_{i}\left( \Phi \right) /\partial \Phi _{k}$ are
given by expression~(\ref{expressionDaDb}) in lemma~(\ref{lem:GradPotentials}%
).
\end{theorem}

\begin{proof}[Proof of theorem~\protect\ref{thm:LLGradient}]
The log-likelihood given in equation~(\ref{eqlogLMEmp}) is made of two
terms, the first of which, $\log \hat{L}_{1}\left( \theta \right) $ only
depends on $\theta $ through $\Phi $, while the second one, $\log \hat{L}%
_{2}\left( \theta \right) $ depends on all the parameters of the model. The
differentiations yielding points (i)-(v) are straightforward.
\end{proof}

\newpage
Supplement to ``A Note on the Estimation of Job Amenities and Labor Productivity.''

This supplement contains two additional sections. The first presents results on how to deal with missing data on transfers, whereas the second introduces the associated concentrated maximum likelihood function.

\section{Extension to randomly missing transfers}

In some applications, data will come from surveys where typically non
response to questions about earnings are frequently encountered. Our
proposed estimation strategy extends to the case where, for some random
matches, transfers are missing. The log-likelihood expression presented in
theorem~\ref{th:LogLikeSample} offers a very intuitive way of understanding
how missing transfers for some random observations will impact the
estimation. To formalize ideas, let $p$ be the probability that for any
arbitrary match the transfer is missing. The sample is still representative
of the population of matches, but a random part of the sample consists of
matches with observed transfers , i.e. $\left( X_{i},Y_{i},W_{i}\right)
_{i=1}^{n^{o}}$, and the other part of matches with missing transfers, i.e. $%
\left( X_{i},Y_{i},.\right) _{i=n^{o}+1}^{n}$ where $n^{o}$ is the number of
matches with observed transfers and $n$ is as before the size of our sample
of matches (we have re-ordered the observations such that those matches with
observed transfers are indexed first). The log-likelihood in this situation
is therefore%
\begin{equation*}
\log \hat{L}\left( \theta \right) =\log \hat{L}_{1}\left( \theta \right)
+\log \hat{L}_{2}\left( \theta \right) +n^{o}\log p+\left( n-n^{o}\right)
\log \left( 1-p\right)
\end{equation*}%
where $\log \hat{L}_{1}\left( \theta \right) $ is given as in equation~(\ref%
{ExpLL1}) and $\log \hat{L}_{2}\left( \theta \right) $ reads now as%
\begin{equation}
\log \hat{L}_{2}\left( \theta \right) =-\sum_{i=1}^{n^{o}}\frac{\left(
W_{i}-w_{i}\left( \theta \right) \right) ^{2}}{2s^{2}}-\frac{n^{o}}{2}\log
s^{2}  \label{ExpLL2missing}
\end{equation}%
thus $p=n^{0}/n$. As $n^{o}$ tends to $0$, and hence $p$ tends to $0$, the
log-likelihood function tends to $\log \hat{L}_{1}\left( \theta \right) $.
In contrast, when $n^{o}$ tends to $n$, and hence $p$ tends to 1, the
expression of $\log \hat{L}_{2}\left( \theta \right) $ in equation~(\ref%
{ExpLL2missing}) tends to that of $\log \hat{L}_{2}\left( \theta \right) $
in equation~(\ref{ExpLL2}) such that the log-likelihood function tends to
equation~(\ref{eqlogLMEmp}).

\section{Concentrated Likelihood}\label{rk:concentratedLL}
In most applications, the parameters of primary
interest are those governing workers' deterministic values of amenities and
firms' deterministic values of productivity, i.e. $A$ and $\Gamma $
respectively. The remaining parameters $\left( \sigma _{1},\sigma
_{2},t,s^{2}\right) $ are auxiliary, and the focus of attention is the \emph{%
concentrated log-likelihood}, which is given by
\begin{equation*}
\log l\left( A,\Gamma \right) :=\max_{\sigma _{1},\sigma _{2},t,s^{2}}\log
\hat{L}\left( \theta \right) =\log \hat{L}_{1}\left( \Phi \right)
+\max_{\sigma _{1},\sigma _{2},t,s^{2}}\log \hat{L}_{2}\left( A,\Gamma
,\sigma _{1},\sigma _{2},t,s^{2}\right) .
\end{equation*}%
where as usual, $\Phi =A+\Gamma $. Denoting $\sigma _{1}^{\ast },\sigma
_{2}^{\ast },t^{\ast }$ and $s^{\ast 2}\,$\ the optimal value of the
corresponding parameters given $A$ and $\Gamma $, one gets
\begin{equation}
\left( \sigma _{1}^{\ast },\sigma _{2}^{\ast },t^{\ast }\right) =\arg
\min_{\sigma _{1},\sigma _{2},t}\sum_{i=1}^{n}\left( W_{i}-w_{i}\left(
\theta \right) \right) ^{2},
\end{equation}%
which is the solution to a Nonlinear Least Squares problem which is readily
implemented in standard statistical packages, and $s^{\ast
2}=n^{-1}\sum_{i=1}^{n}\left( W_{i}-w_{i}\left( \theta ^{\ast }\right)
\right) ^{2}$. The partial derivative of the concentrated log-likelihood
with respect to $A_{k}$ is given by%
\begin{equation*}
\frac{\partial \log l\left( A,\Gamma \right) }{\partial A_{k}}=\frac{%
\partial \log \hat{L}_{1}\left( \Phi \right) }{\partial \Phi _{k}}+\frac{%
\partial \log \hat{L}_{2}\left( A,\Gamma ,\sigma _{1}^{\ast },\sigma
_{2}^{\ast },t^{\ast },s^{\ast 2}\right) }{\partial A_{k}}
\end{equation*}%
and a similar expression holds for $\partial \log l/\partial \Gamma _{k}$.
These formulas are derived in the following proof.

\begin{proof}
Recall $\theta =\left( A,\Gamma ,\sigma _{1},\sigma _{2},t,s^{2}\right) $.
The maximum likelihood problem can be written as%
\begin{equation*}
\max_{\theta }\log \hat{L}\left( \theta \right) =\max_{A,\Gamma }\log
l\left( A,\Gamma \right)
\end{equation*}%
where $\log l\left( A,\Gamma \right) =\max_{\sigma _{1},\sigma
_{2},t,s^{2}}\log \hat{L}\left( \theta \right) $ is the concentrated
log-likelihood which can be rewritten as
\begin{equation}
\log l\left( A,\Gamma \right) =\log \hat{L}_{1}\left( \theta \right)
+\max_{\sigma _{1},\sigma _{2},t,s^{2}}\log \hat{L}_{2}\left( \theta \right)
.  \label{ConLogL}
\end{equation}%
where
\begin{equation}
\max_{\sigma _{1},\sigma _{2},t,s^{2}}\log \hat{L}_{2}\left( \theta \right)
=-\min_{s^{2}}\left( \frac{n}{2}\log s^{2}+\frac{1}{2s^{2}}\min_{\sigma
_{1},\sigma _{2},t}\sum_{i=1}^{n}\left( W_{i}-w_{i}\left( \theta \right)
\right) ^{2}\right)  \label{eqmaxCondL2}
\end{equation}%
The second minimization in equation~(\ref{eqmaxCondL2}) is an Ordinary Least
Squares problem whose solution given $A,\Gamma $, denoted $\left( \sigma
_{1}^{\ast },\sigma _{2}^{\ast },t^{\ast }\right) $, is the vector of
coefficients of the OLS\ regression of $W$ on $\left( \gamma -b,a-\alpha
,1\right) $. The value of $s^{2}$, denoted $s^{\ast 2}$, is given by%
\begin{equation*}
s^{\ast 2}=\frac{\sum_{i=1}^{n}\left( W_{i}-w_{i}\left( \theta ^{\ast
}\right) \right) ^{2}}{n}.
\end{equation*}

The envelope theorem yields an expression for the gradient of the
concentrated log-likelihood with respect to the concentrated parameters $A$
and $\Gamma $, that is%
\begin{equation*}
\nabla _{A,\Gamma }\log l\left( A,\Gamma \right) =\nabla _{A,\Gamma }\log
\hat{L}_{1}\left( \theta ^{\ast }\right) +\nabla _{A,\Gamma }\log \hat{L}%
_{2}\left( \theta ^{\ast }\right) .
\end{equation*}%
The elements of the first part of the gradient are given in theorem~\ref%
{thm:LLGradient} part (i) whereas parts (ii), (iv) and (v) of theorem~\ref%
{thm:LLGradient} provide the building blocks for the elements of the second
part of the gradient.
\end{proof}

\end{document}